\documentclass[sn-mathphys-ay]{sn-jnl}

\usepackage{graphicx}%
\usepackage{multirow}%
\usepackage{amsmath,amssymb,amsfonts}%
\usepackage{amsthm}%
\usepackage{mathrsfs}%
\usepackage[title]{appendix}%
\usepackage{xcolor}%
\usepackage{textcomp}%
\usepackage{manyfoot}%
\usepackage{booktabs}%
\usepackage{algorithm}%
\usepackage{algorithmicx}%
\usepackage{algpseudocode}%
\usepackage{listings}%
\usepackage{comment}

\usepackage[utf8]{inputenc}
\usepackage{bm}
\usepackage{hyperref}
\hypersetup{
	colorlinks=True,       
	linkcolor=blue,          
	citecolor=blue,        
	filecolor=magenta,      
	urlcolor=blue           
}

\newcommand{\rvline}{\hspace*{-\arraycolsep}\vline\hspace*{-\arraycolsep}}

\newcommand{\R}{\mathbb{R}}

\newcommand{\N}{\mathbb{N}}
\newcommand{\I}{\mathbb{I}}

\newcommand{\bc}{\mathbf{c}}
\newcommand{\bx}{\mathbf{x}}
\newcommand{\bv}{\bm{v}}
\newcommand{\bw}{\bm{w}}
\newcommand{\be}{\mathbf{e}}

\newcommand{\cA}{\mathcal{A}}

\def\nn{\nonumber}

\theoremstyle{thmstyleone}%
\newtheorem{theorem}{Theorem}
\theoremstyle{thmstyletwo}%
\newtheorem{lemma}[theorem]{Lemma}

\newtheorem{counterexample}[theorem]{Counterexample}
\newtheorem{corollary}[theorem]{Corollary}

\theoremstyle{thmstylethree}%
\newtheorem{definition}{Definition}%

\begin{document}
	
	\title[PageRank and the control of a random walker]{Fully personalized PageRank and algebraic methods to distribute a random walker}
	
	
	\author*[1,2]{\fnm{Gonzalo} \sur{Contreras-Aso}}\email{gonzalo.contreras@urjc.es}
	
	\author[1,2,3]{\fnm{Regino} \sur{Criado}}\email{regino.criado@urjc.es}
	
	\author[1,2,3]{\fnm{Miguel} \sur{Romance}}\email{miguel.romance@urjc.es}
	
	\affil*[1]{\orgdiv{Departamento de Matem\'atica Aplicada, Ciencia e Ingenier\'ia de los Materiales y Tecnolog\'ia Electr\'onica}, \orgname{Universidad Rey Juan Carlos}, \orgaddress{\street{c/Tulip\'an s/n}, \city{M\'ostoles}, \postcode{28933}, \state{Madrid}, \country{Spain}}}
	
	\affil[2]{\orgdiv{Laboratory of Mathematical Computation on Complex Networks and their Applications}, \orgname{Universidad Rey Juan Carlos},  \orgaddress{\street{c/Tulip\'an s/n}, \city{M\'ostoles}, \postcode{28933}, \state{Madrid}, \country{Spain}}}
	
	\affil[3]{\orgdiv{Data, Complex networks and Cybersecurity Research Institute}, \orgname{Universidad Rey Juan Carlos}, \orgaddress{\street{Plaza Manuel Becerra 14}, \city{Madrid}, \postcode{28028}, \state{Madrid}, \country{Spain}}}
	
	
	\abstract{We present a comprehensive analysis of algebraic methods for controlling the stationary distribution of PageRank-like random walkers. Building upon existing literature, we compile and extend results regarding both structural control (through network modifications) and parametric control (through measure parameters) of these centralities. We characterize the conditions for complete control of centrality scores and the weaker notion of ranking control, establishing bounds for the required parameters. Our analysis includes classical PageRank alongside two generalizations: node-dependent dampings and node-dependent personalization vector, with the latter being a novel idea in the literature. We examine how their underlying random walk structures affect their controllability, and we also investigate the concepts of competitors and leaders in centrality rankings, providing insights into how parameter variations can influence node importance hierarchies. These results advance our understanding of the interplay between algebraic control and stochastic dynamics in network centrality measures.}
	
	
	
	
	\maketitle



\section{Introduction}\label{sec:introduction}

Network centrality measures have emerged as fundamental tools in the analysis of complex systems, providing crucial insights into node importance across diverse applications ranging from social network analysis to biological systems \cite{boccaletti2006complex, newman2010networks}. Among these measures, spectral centralities \cite{vigna2016spectral}- and particularly PageRank \cite{page1998pagerank} and its variants - hold special significance due to their deep connection to stochastic processes, while still retaining algebraic tractability. PageRank can be interpreted as the stationary distribution of a random walk on a graph, where a ``random surfer'' moves through the network following edges with probability $\alpha$ and teleports to random nodes with probability 1-$\alpha$. This interpretation through the lens of Markov processes provides both theoretical elegance and practical robustness to the measure \cite{langville2006google}.

The stochastic nature of PageRank raises intriguing questions about the controllability and manipulation of such measures. How much can we influence the outcome of a random process while maintaining its essential probabilistic character? This tension between deterministic control and stochastic behavior lies at the heart of our investigation. In this paper, we conduct a comprehensive study of the controllability of PageRank scores and rankings through both structural modifications of the underlying network and parametric adjustments of the centrality measures themselves.

Our analysis encompasses not only the classical PageRank algorithm but also two important variants, which introduce personalized rules for the random walk at each node: on the one hand, we consider the node-dependent damping discussed in \cite{avrachenkov2014personalized}, which allows for heterogeneous teleportation probabilities across the network; on the other hand we consider a novel notion of personalization in PageRank, considering a node-dependent personalization vector, something equivalent to using a ``personalization matrix''. These variants introduce additional degrees of freedom in the underlying stochastic process, potentially affecting the controllability of the resulting centrality measures. Through careful application of algebraic and analytical techniques, we establish bounds on parameter regions that enable various degrees of control over these measures. We differentiate between localization (the ranges of the possible centralities attainable for each node for a fixed damping factor) \cite{garcia2012localization}, complete control (the ability to achieve specific centrality scores for all nodes) and the weaker notion of ranking control (the ordered relationship between nodes) \cite{contrerasaso2023pagerank}.

This investigation not only advances our theoretical understanding of spectral centralities but also has practical implications for applications where influencing node importance is crucial, from search engine optimization \cite{chaffey2009search} to even  biology \cite{morrison2005generank,gleich2015pagerank}. By examining how the stochastic foundation of these measures interacts with our ability to control their outcomes, we provide insights into both the possibilities and limitations of centrality manipulation in complex networks.

This article is structured as follows: we begin with a review of some basic concepts of spectral centrality measures and their relation to Markovian processes in Section~\ref{sec:preliminary}. We then discuss some novel results of structural controllability in PageRank, mainly weight tuning and self-loop score increases, in Section~\ref{sec:structcontrol}. Section~\ref{sec:paramcontrol} provides an overview of some foundational results in the parametric controllability of PageRank (localization, complete and ranking control), which we then extend to some of its generalizations. We finish with some conclusions and outlook.

\paragraph{A note on \textit{Search Engine Optimization}}\label{subsec:SEO}
Something which we have not studied is the possibility of adding new links among existing nodes, however this is a very realistic possibility in human-made networks such as the Internet: a webpage can choose which hyperlinks it creates, bridging it to other webpages. As a matter of fact, this problem has been thoroughly studied by the Computer Science community, to the point where its practice has been dubbed \textit{Search Engine Optimization} or SEO \cite{langville2006google}, for short. 

Studies in this area exploded in popularity during the first decade of the 2000's, after PageRank made its appearance and thanks to it Google became the predominant search engine, and before Google moved on to more sophisticated and modified versions of the algorithm (with several link spam filters and using Machine Learning, for instance). Plenty of websites tried to attract as much attention as possible, and therefore needed to be highly positioned in Google's ranking. There were many attempts to mischievously exploit the algorithm via link farms or the so-called ``Google bombs'' \cite{bar2007google, langville2006google}, but there was also a vast and genuine attempt to understand thoroughly and quantitatively how did changes affect the PageRank centrality. For the interested reader, we refer them to \cite{avrachenkov2006effect, olsen2010maximizing, sheldon2010manipulation, fercoq2013ergodic, csaji2014pagerank}  and the references therein.

\section{Preliminary notions}\label{sec:preliminary}

Centrality measures are fundamental tools in network analysis, used to identify the most influential nodes in a graph. These measures quantify a node's importance based on its position and relationships within the network structure. Centrality can be interpreted in various ways, leading to diverse metrics such as degree centrality, and betweenness centrality, each capturing different aspects of influence.

Among these, spectral centrality measures leverage the eigenvalues and eigenvectors of matrices representing the network, most notably the adjacency matrix. Spectral centralities are particularly interesting for their mathematical tractability as well as for their computational efficiency, as compared to, for instance, the betweenness centrality.

In this section we are going to briefly summarize the most relevant spectral centralities in graphs, with a particular focus on PageRank and two of its variants, which can be understood through the lens of both linear algebra as well as stochastic processes.

\subsection{Spectral centrality measures in graphs}

Let $G=(V,E)$ be a graph with node set $V$ and edge set $E\subseteq V\times V$. We denote as $N=|V|$ the number of nodes, $L=|E|$ the number of edges. If for each edge $(i,j)\in E$ $(j,i)\in E$, then the graph is undirected, otherwise directed. We will often consider weighted graphs as well, i.e. we endow the graph with a function $w:E\rightarrow \R$. For reasons which will be clear soon, we will only consider positive weightings, $w(i,j)>0$. Lastly, we will at one point consider multigraphs as well: these can be thought of as weighted graphs with $w\in \N$.

Some graph theoretical concepts will be necessary in what follows. Some of those are quite well-known, we will briefly go over them now. We will denote as $k_i\in\N$ the (weighted) degree of node $i$, the number of neighbors. In the case of directed graphs each node $i$ will have both in-degree $k^{\rm in}_i$ and $k^{\rm out}_i$. The adjacency matrix $A=(a_{ij})\in\R^{N\times N}$ is defined as
\begin{equation}
	a_{ij} = \begin{cases}
		1 & \text{if } (i,j) \in E,\\
		0 & \text{otherwise.}
	\end{cases}
\end{equation}

The generalization to weighted graphs is straightforward. The adjacency matrix is intrinsically related to the connectivity of the graph: a (directed) graph is (strongly) connected, if and only if its adjacency matrix is irreducible 
\cite{meyer2001matrix}. This will play a key role when we define centralities based on the adjacency matrix.

A very active area of research in the last decades is the study of centralities in graphs \cite{saxena2020centrality}. A centrality measure $f:V\longrightarrow \R $ assigns a score to each node in the graph, quantifying their importance with respect to a give heuristic or metric. These scores have to be positive, and the overall assignment unique. Among all centrality measures, spectral ones \cite{vigna2016spectral} are particularly important due to their analytical properties and efficient computation, with PageRank \cite{page1998pagerank} being in the forefront due to its foundational role in the Google search engine \cite{langville2003deeper}.

\subsubsection{Perron-Frobenius theory}

In general when discussing spectral centralities of nodes in graphs, it is sensible to only consider strongly connected graphs: the influence of any node must be able to reach any other node. As we just mentioned, strongly connected graphs are represented by irreducible adjacency matrices. The Perron-Frobenius Theorem provides analytical guarantees for the existence of a unique, positive eigenvector of such matrices.

\begin{theorem}[Perron-Frobenius \cite{perron1907theorie, frobenius1912matrizen}] \label{thm:Perron-Frobenius}
	Let $A\in \R^{N\times N}$ be a non-negative, irreducible square matrix. Then the following statements are true:
	\begin{enumerate}
		\item $\rho(A)\in \sigma(A)$ and its algebraic multiplicity is 1.
		\item There exists a vector $\bc\in\R^n$ with $\bc>0$ such that it is an eigenvector of $A$ associated to the eigenvalue $\rho(A)$.
		\item The eigenvector $\bc$ is unique up to scaling.
	\end{enumerate}
	The eigenvector $\bc$ is sometimes referred to as the Perron eigenvector of $A$.    
\end{theorem}

It is important to note that there is also the preceding Perron Theorem \cite{perron1907theorie}, which is analogous and applies to positive square matrices without the need to consider irreducibility. This is relevant in the case of PageRank (where the addition of the personalization part renders the Google matrix positive).

We are now equipped to discuss spectral centrality measures in graphs.

\subsubsection{The eigenvector centrality} 

The quintessential spectral centrality comes from considering that not only the amount of acquaintances matter (as in the degree centrality) but also the importance of each of them. Formalizing this idea for the importance $c(i)$ of a node \cite{bonacich1972technique}, we have
\begin{equation}\label{eq:eigcent_derivation}
	c(i) \propto \sum_{j\rightarrow i} c(j) \quad \Rightarrow\quad  \lambda \bc = A^T \bc,
\end{equation}
where we have introduced the proportionality constant $\lambda$, and the adjacency matrix $A = (a_{ij})$ to sum over neighbors. This is recognized as the eigenvector equation of the transposed adjacency matrix.

If a graph is (strongly) connected, then its adjacency matrix is irreducible, and due to the Perron-Frobenius Theorem \cite{perron1907theorie} there is a unique, positive eigenvector $\bc>0$, which is associated to the spectral radius $\rho(A)$. Said eigenvector fulfills the requirements of a centrality measure (being positive and unique), hence it is referred to as the eigenvector centrality of the graph.

\subsubsection{PageRank} 

In 1998, Sergei Brin and Larry Page proposed an algorithm to rank webpages \cite{page1998pagerank}, later implemented as the base of their newly created search engine, Google. Their idea was to formulate a mathematical model for a typical Internet surfer as a stochastic process on a network: namely, they consider a Markov chain (i.e. a random walker) with the following transition rules: starting at node $i$ 
\begin{itemize}
	\item with probability $\alpha\in[0,1]$, it traverses one of the out-edges. The probability of choosing an out-edge is proportional to its weight.
	
	\item With probability $1-\alpha$, it jumps elsewhere at random, with the choice being dictated by the probability distribution $\bv\in\R^N$.
\end{itemize}
This stochastic process would be iterated \textit{ad infinitum}. This was supposed to model the behavior of the average Internet surfer, following hyperlinks until at some point they would go to a different webpage altogether. 

The resulting stationary distribution (how many times each node is visited in the limit $t\rightarrow\infty$) would be later recognized by the network science community as a centrality measure intimately related to the eigenvector centrality, but for a modified version of the original adjacency matrix. This stationary distribution is used at Google as a proxy for the importance of each webpage \cite{page1998pagerank}. 

The PageRank centrality of a graph is the unique, positive, left eigenvector $\bm{\pi}(\alpha,\bm{v})\in\R^N$ of the so-called ``Google matrix'' $G$
\begin{equation}\label{eq:PageRank}
	G(\alpha,\bv) = \alpha P + (1-\alpha) \be \cdot \bv^T, \qquad p_{ij}=\frac{a_{ij}}{\sum_j a_{ij}},
\end{equation}
where $P=(p_{ij})$ is the row-normalized adjacency matrix and $\be=(1,\dots,1)^T$. The two parameters $\alpha\in[0,1]$ (damping factor), and $\bv\in\R^N,\, |v|_1=1$ (personalization vector) must be set beforehand, and their interplay will be studied later in this manuscript.

There are some facts which are important to note about this measure: first, the Perron-Frobenius Theorem \cite{perron1907theorie} again guarantees the existence and uniqueness of the PageRank vector\footnote{If there are ``dangling nodes'', i.e. nodes with zero out-degree, these are problematic from an irreducibility point of view. This is often solved by substituting $P\rightarrow P + \mathbf{d} \cdot \bm{u}^T$, where $\bm{d}=(d_i)\in\R^N$ is the vector indicating the dangling nodes ($d_i=1$ if $i$ is a dangling node, otherwise $d_i = 0$) and $\bm{u}\in\R^N,\,||\bm{u}||_1$ is the dangling node distribution.}, and second, this vector is unit norm, $|\pi|_1 =\bm{\pi}^T \be = 1$, by construction.

Using this last fact we can obtain an explicit formula for the PageRank vector \cite{boldi2005pagerank}
\begin{equation}\label{eq:PRexplicit}
	\bm{\pi}^T = (1-\alpha)\bv^T(1-\alpha P)^{-1}= (1-\alpha)\bv^T \sum_{n=0}^\infty \alpha^n P^n,
\end{equation}
where the last equality defines essentially the power method computation of the PageRank vector, using the personalization vector as the seed of the power iteration. This is behind the usefulness of PageRank, as the power method is simple to implement in practice due to its low computational complexity.

\paragraph{Node-dependent restart PageRank} 

In \cite{avrachenkov2014personalized} a different generalization of PageRank was put forward, where they consider node-dependent dampings $\alpha_i\in(0,1),\, i\in 1,\dots,N$. This generalizes the stochastic process to the situation where the chance of following the links or teleporting elsewhere depends on where the random walker is located at present in the graph.

The node-dependent restart PageRank vector $\bm{\pi}_{\rm NPR}(\alpha_1, ..., \alpha_N, \bw)\in \R^{N}$ is defined as the positive, leading eigenvector $\bm{\pi}_{\rm NPR}\in\mathbb{R}^N$ satisfying
\begin{equation}\label{eq:NDRPageRank}
	\bm{\pi}_{\rm NPR}^T \left[\mathcal{A} P + (\I_N-\cA) \be \cdot  \bw^T\right] = \bm{\pi}_{\rm NPR}^T,\quad \cA = \text{diag}(\alpha_1, ..., \alpha_N),
\end{equation}
where $\cA=\text{diag}(\alpha_1, ..., \alpha_N) \in \R^{N\times N}$ is the diagonal matrix containing each damping and $\bw\in\R^N$ is again a personalization vector $\bw>0,||\bw||_1=1$.

The introduction of this node-dependent damping $\alpha_i$ enables the discussion of different possibilities and cases of interest, like the usual PageRank ($\alpha_i=\alpha,\,\forall i$) or the degree dependent restart $\cA = \I_N - a D^\sigma$ where $D$ is the (out)-degree matrix and $a, \sigma$ are tunable parameters. This choice is realistic from an Internet surfer point of view (choosing to follow links or jump elsewhere based on the amount of links available), which is why in the original paper \cite{avrachenkov2014personalized} they entertain this possibility.

\paragraph{Personalization matrix PageRank}

A matter that has been always overlooked in the literature is the possibility to include a personalization matrix instead of a personalization vector. In the original PageRank, the teleportation part of the random walk is dictated by the matrix $\be\cdot \bv^T\in\R^{N\times N}$. This row-stochastic matrix has identical rows, which can be interpreted as the random walker finding the same probability to teleport to other nodes regardless of where it teleports from. 

Following similar heuristics as the one from the node-dependent damping, we can consider a ``node-dependent personalization vector'', where rows need not be identical. We end up with the following expression for the personalization matrix PageRank $\bm{\pi}_{\rm MPR}(\alpha,M)\in \R^{N}$
\begin{equation}\label{eq:PMPageRank}
	\bm{\pi}_{\rm MPR}^T \left[\alpha P + (1-\alpha)  M\right] = \bm{\pi}_{\rm MPR}^T,
\end{equation}
where $M=(m_{ij})\in\R^{N\times N}$ is a row-stochastic matrix. The component $m_{ij}$ indicates the probability of teleporting to node $j$, once the random walker located at node $i$ has chosen to teleport. 

To the author's knowledge, this generalization of PageRank has never been discussed in the literature, even though it is rather sensible. One possible explanation is the lack of an explicit formula such as \eqref{eq:PRexplicit}. Another reason will be discussed in Subsection~\ref{subsec:persmatrix}.

One could also consider a completely personalized PageRank by the node-dependent damping and the node-dependent personalization. For the sake of conciseness we will not discuss this, as discussing the two cases separately will be enough for our purposes.

\section{Structural controllability}\label{sec:structcontrol}

In the previous section we introduced the PageRank centrality measure, which can be interpreted as a stochastic process on a network, as well as some related centralities. This stochastic processes clearly depend on the underlying network properties: the simplest example of this is the connectivity of the network, a random walker in a disconnected graph will never traverse it entirely. It is therefore important to understand how much do changes in the structure of the graph affect the centrality outcome, i.e., how ``structurally controllable'' these measures are.


\subsection{Weight tuning control}\label{subsec:weighttune}

One of the simplest structural changes one can consider is altering the weight of existing edges in the graph. If we start from a strongly connected graph, as long as we maintain the positivity of the edges, we should still be able to use the analytical guarantees from the Perron-Frobenius Theorem \ref{thm:Perron-Frobenius}.

Following this research idea, in \cite{latora2012controlling} the authors carried out a complete characterization of the relation between changing the weights of a subset of edges and changing the centrality score of each node.

The most important result in the aforementioned paper is the following theorem:

\begin{theorem}\label{thm:weightcontrol}
	Let $G=(V, E)$ be a directed, possibly weighted, strongly connected graph with $N=|V|$ nodes, let $\mathbf{c}\in \R^N$ be a positive vector. It is always possible to assign weights $w_{ij}$ to edges $(i,j)\in E$ such that the eigenvector centrality of $G$ is $\mathbf{c}$.
\end{theorem}

This is a very powerful result: no matter the internal connectivity of a directed network, as long as it is strongly connected we can always find some weighting of its edges such that the eigenvector centrality of the weighted graph is one of our choosing. An example of this can be seen in Figure~\ref{fig:fullcontrol}.

\begin{figure}[h]
	\centering
	\includegraphics[width=0.8\textwidth]{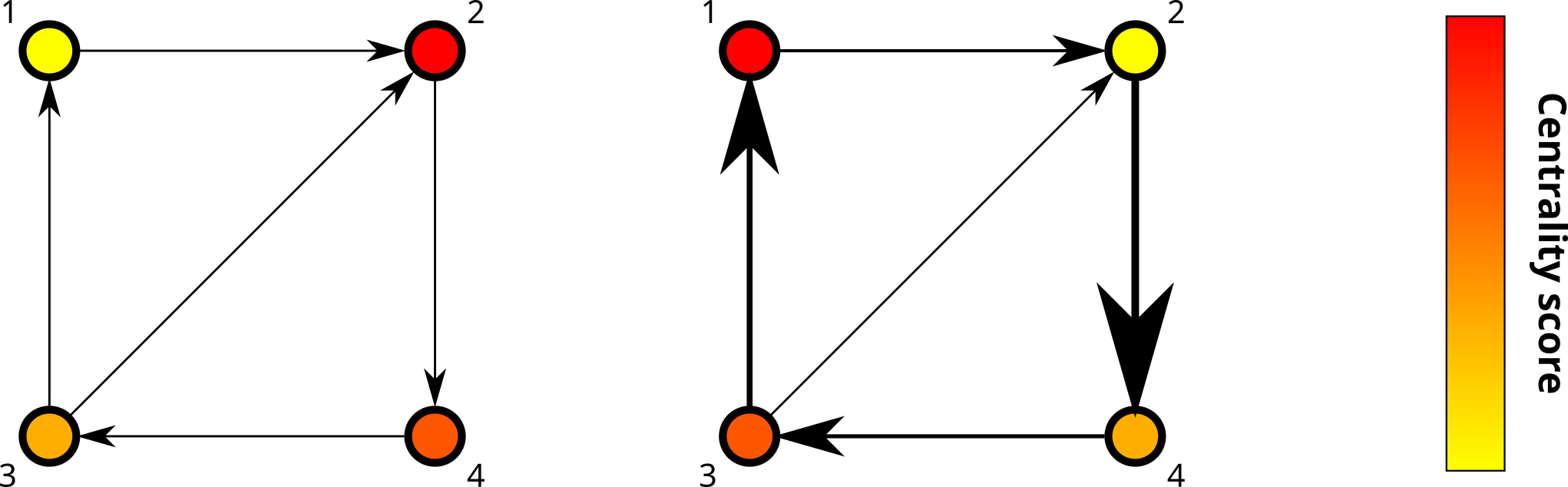}
	\caption[Toy example of a graph whose centrality has been modified via weight changes.]{A toy example of a directed, strongly connected graph with 4 nodes and 5 edges, whose centrality ranking has been modified to be the opposite of the original one via weight changes, represented as thicker arrows.}
	\label{fig:fullcontrol}
\end{figure}

Perhaps more interestingly then, is the question of finding the minimum subset of edges whose weights need to be changed for any desired ranking to be achieved. This question was also addressed in \cite{latora2012controlling}, but rather than considering the edges whose weights need to be changed, they considered a related problem, which is finding the Minimum Controlling Subset (MCS) of nodes, whose edges need to have different weights. Surprisingly, they found that in many cases of real networks, the relative size of the MCS compared to the size of the full network is quite small, around 10\% of it or even less.

It should be noted that for this method to work it is essential that the graph be directed: undirected graphs have the same weight in both directions, and that restriction of the degrees of freedom generally renders this method useless. As a simple example of this, consider two nodes joined by an undirected, weighted edge: their centrality score is the same regardless of the weight.

\subsubsection{Weight tuning in PageRank and its variants}

The deep relation between the eigenvector centrality and PageRank, mediated by the Perron-Frobenius Theorem \ref{thm:Perron-Frobenius} make it quite tantalizing the possibility of applying the same method which worked in the directed graph case to PageRank. Sadly, as we will clearly see there is a fundamental problem in this scenario, which renders PageRank generally uncontrollable using this method even in the directed graph case.

The reason for this is the row-normalization of the adjacency matrix: the construction of $P$ normalizes out any weight placed on out-edges coming from nodes with out-degree equal to 1. The simplest way to see this is considering directed rings, as in the following counterexample.

\begin{counterexample}
	Consider the $N=6$ directed cycle $G=(V,E)$ from Figure~\ref{fig:ring-graph}. Its adjacency matrix is of the form $a_{ij}=\delta_{j,(i+1 \mod 6)}$. Each node clearly has $k_i^{\rm in} = k_i^{\rm out} = 1,\, \forall i$. If we modify the weight of edge $(i,j)\in E$ to be $w_{ij}>0$, then $k_i^{out}=w_{ij}$ and it therefore translates in the row-adjacency matrix $P$ of PageRank back to $p_{ij}=1$, regardless of $w_{ij}$. 
	
	\begin{figure}[h]
		\centering
		\includegraphics[width=0.6\textwidth]{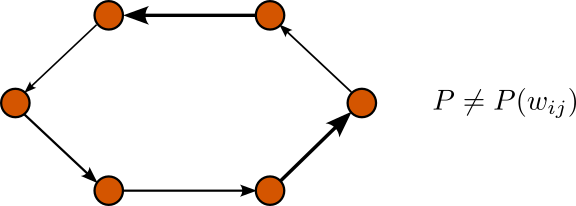}
		\caption[Simple example of a network whose PageRank centrality is unaffected by any modification of the edge weights]{Simple example of a network (the directed cycle $C_6$) whose PageRank centrality is unaffected by any modification of the edge weights.}
		\label{fig:ring-graph}
	\end{figure}
\end{counterexample}

This counterexample applies to all PageRank variants introduced in the previous section: it is an impossibility coming from the row-stochasticity of $P$, which is present in all equations.

\subsection{Self-loop tuning in PageRank}

Another structural change which we could consider is the addition of a single self-loop over a node, with arbitrary weight $w$, in the hopes that said node will become the most central. From the stochastic process point of view, this seems rather sensible: the random walker will likely choose to traverse the self-loop over and over again, thus increasing said node's PageRank score.

This is, however, more nuanced: increasing the self-loop weight will essentially diminish the weights of other edges, coming out of that node, but the row-normalization will not allow its weight to be much larger than that of the other edges in the graph. Moreover, for low values of $\alpha$ the walker will likely not choose to traverse any edge (hence, ignoring the self loop), but to jump at random to another node based on the personalization vector $\bv$. Nevertheless, we can establish a bound for the values of $\alpha$ which admit a PageRank centrality where the node with a self-loop has the highest score.

Before doing so, let us enunciate the following important result
\begin{theorem}(Theorem 8, p. 130 in \cite{lax2007linear}) \label{thm:lax}
	Let $A(t)$ be a differentiable matrix-valued function of $t$, $a(t)$ an eigenvalue of $A(t)$ of multiplicity one. Then we can choose an eigenvector $\bm{h}(t)$ of $A(t)$ pertaining to the eigenvalue $a(t)$ to depend differentiably on~$t$.
\end{theorem}

This will be a key result in the upcoming theorem, which shows that a node can increase its centrality score beyond that of the rest of the nodes after increasing its self-weight and the damping factor high enough.

\begin{theorem}[Self-loop PageRank increase]
	Let $G=(V,E)$ be a strongly connected graph, with row-normalized adjacency matrix $P$. If we allow a node to establish an arbitrary number of self-loops onto itself (alternatively, a single self-loop with arbitrarily high weight), then for $1-1/w \leq \alpha \leq 1$ that node can always achieve the highest PageRank value. 
\end{theorem}

\begin{proof}
	Without loss of generality, we choose node 1 to be that with a tunable self-loop weight. For simplicity, we first consider the $\alpha=1$ (no personalization) case, we later discuss what changes when $\alpha \neq 0$.

	In that scenario, we need to compute the Perron eigenvector of $P$. The row normalization yields the following values
	\begin{equation}
		P_{11}=\frac{w}{k_1^{out}+w}, \quad P_{1j} = \frac{a_{1j}}{k_1^{out}+w}
		\quad \Rightarrow \quad P(w)=
		{
			\renewcommand*{\arraystretch}{2}
			\begin{pmatrix}
				\dfrac{w}{k_1^{out} + w} & \rvline & \dfrac{a_{1j}}{k_1^{out}+w} \\
				\hline
				\mathcal{O}(1) & \rvline & \mathcal{O}(1)
			\end{pmatrix},
		}
	\end{equation}
	where $w$ is the weight of the self-loop.
	
	Due to the strong connectedness, for finite $w$ we have that $P(w)$ is an irreducible matrix, therefore it has a unique, positive eigenvector $\bm{\pi}$ associated to the spectral radius $\lambda=1$. 
	
	In the limit we have
	\begin{equation}
		P_\infty=\lim_{w\rightarrow\infty} P(w) \Rightarrow 
		P_\infty^T = 
		{
			\renewcommand*{\arraystretch}{1.5}
			\begin{pmatrix}
				1 & \rvline & \mathcal{O}(1) \\
				\hline
				\bm{0} & \rvline & \mathcal{O}(1)
			\end{pmatrix}\equiv
			\begin{pmatrix}
				1 & \rvline & \bc \\
				\hline
				\bm{0} & \rvline & B
			\end{pmatrix}, 
		}
	\end{equation}
	where the respective submatrices are
	\begin{equation}
		\bm{0}=
		(0,\dots, 0)^T
		\in\R^{(N-1)\times 1}, \quad \bc\in\R^{1\times(N-1)}, \quad B\in\R^{(N-1)\times(N-1)}.
	\end{equation}
	
	$P_\infty^T$ possesses $\be_1=(1,0,\dots,0)^T$ as eigenvector with eigenvalue $\lambda=1$. The remaining eigenvectors are orthogonal to this one, and their non-trivial part comes from those of the submatrix $B=(b_{ij})\in\R^{(N-1)\times(N-1)}$. This submatrix is irreducible, and thus the Perron-Frobenius Theorem \ref{thm:Perron-Frobenius} applies to it, in particular we have the following property for its spectral radius $r_B$ \cite{meyer2001matrix}
	\begin{equation}
		r_B \leq \max_i\sum_{j} b_{ij},
	\end{equation}
	and due to the fact that $B$ is no longer column stochastic (as was $P_\infty^T$), we have $r_B<1$. Therefore, $\lambda = 1$ is an eigenvalue of $P(w)$ with algebraic multiplicity 1 for all $w\in(0,\infty)$. Furthermore, $P(w)$ is differentiable in that range, hence by Theorem~\ref{thm:lax} its associated eigenvector is continuous, and therefore the only non-negligible component of the eigenvector for high enough, finite $w$ will be that of the first node.
	
	Lastly, we should consider the effect of decreasing $\alpha$ to less than 1. In that case a necessary (but not sufficient) condition for the first node to still retain the highest importance would be if $\alpha \in (1 - 1/w, 1)$ for $w$ large enough. To see this, notice that the Google matrix \eqref{eq:PageRank} becomes for the minimum $\alpha$
	\begin{equation}
		\mathbb{G} = \alpha P + (1-\alpha)\, \be \cdot \bv^T = \alpha P + \mathcal{O}(w^{-1}) \be \cdot \bv^T. 
	\end{equation}
	Therefore, the leading order contribution is still that coming from $P$, and the previous arguments still apply. 
\end{proof}

The two extreme cases are rather sensible from the random walker perspective: if $w\rightarrow 1$ then $\alpha\rightarrow 0$, in which case the centrality is set by $\bv$, which can indeed dictate that node 1 is the most central (the walker can teleport to it more frequently); if $w\rightarrow\infty$ then $\alpha\rightarrow 1$, in which case as soon as the walker reaches the node, since it no longer teleports it will be forced to use the self-loop over and over again.

Note that the above theorem is still valid in the case of the matrix personalization PageRank \eqref{eq:PMPageRank}, since its modification only affects the part of the Google matrix which is subleading in $w$. The node-dependent restart \eqref{eq:NDRPageRank} generalization presents more challenges in this situation when we move away from trivial parameters (i.e. $\alpha_i\rightarrow1,\,\forall i$, respectively).

\section{Parametric changes}\label{sec:paramcontrol}

When it comes to the controllability of centrality measures, there is sometimes the possibility of modifying them via the parameters which may be involved in it, if any. Clearly, this is not always the case: for instance, the eigenvector centrality is parameter-free. 

On the contrary, PageRank contains not one, but two parameters: the damping factor $\alpha\in[0,1]$ and the personalization vector $\bv\in\R$, $||\bv||_1=1$. We will devote this section to understanding the relation between the centrality outcomes of PageRank (as well as other PageRank-related measures) and the parameters involved, to see how do the latter influence the former.

In order to ease the notation and statements which are to come, we will always be assuming graphs $G=(V,E)$ without dangling nodes (i.e. nodes with zero out-degree, $k_i^{\rm out}=0$), for sensibility of the PageRank measure, something we already commented in Section~\ref{sec:preliminary}.

\subsection{Localization of PageRank, competitors and leaders} \label{subsec:PRlocalization}

A first step towards the characterization of the relation between centralities and parameters in the PageRank centrality measure can be found in \cite{garcia2012localization}, where they associate each node in a network to an interval in the real line which symbolizes the possible values of its centrality score for a fixed damping factor. These intervals are then used to extract information about leaders, followers and competitors in the network. We now briefly review the main points of said article, which we will generalize in Subsection~\ref{subsec:paramRelatedcontrol} for other PageRank-related measures.

First of all, note that for a fixed graph $G$ and damping factor $\alpha\in(0,1)$ the possible centrality vectors $\bm{\pi}$ depend on the choice of personalization vector $\bm{v}$, i.e. $\bm{\pi}=\bm{\pi}(\bm{v})$. 

It is natural to then consider the following definition:
\begin{definition}\label{def:localizationPageRank}
	Let $G=(V,E)$ be a graph with $|V|=N$. The localization of PageRank for node $i$ is the set of all possible values of the PageRank centrality for said node,
	\begin{equation}
		\mathcal{PR}(i)=\{\bm{\pi}^T(\bm{v}) \bm{e}_i \; \forall \bm{v}\in\R^N, \bm{v}>0, ||\bm{v}||_1=1 \},
	\end{equation}
\end{definition}

where $e_i$ is the unit vector in the $i$'th direction. In \cite{garcia2012localization}, we find the following theorem.
\begin{theorem}\label{thm:localizationPageRank}
	Given a graph $G=(V,E)$ and a fixed damping factor $\alpha\in (0,1)$, for each node $i\in \N$
	\begin{equation}
		\mathcal{PR}(i) = (\min_j x_{ji} , x_{ii}),
	\end{equation}
	where $X = (x_{ij}) = (1 - \alpha)(\I_N - \alpha P)^{-1} \in \R^{N\times N}$.
\end{theorem}

This is a rather simple but powerful result: one computes matrix $X$ from the row-normalized adjacency matrix of the graph and the damping factor, and its components constrain the possible values of all centrality scores in the graph, regardless of the personalization vector. This information can be used to analyze the existence of ``effective competitors'' in the network.

It is worth mentioning that the localization of PageRank values has recently been generalized to the case of temporal networks in \cite{aleja2024timedependent}.


\subsection{PageRank rankings and the personalization vector}\label{subsec:paramPRcontrol}

While the previous results are focused at individual centrality scores, a natural step forward is the analysis of the \textit{entire} centrality outcomes as functions of the parameters. This is something we tackled in \cite{contrerasaso2023pagerank}, which we now proceed to discuss.

In order to answer this problem we derived bounds relating the damping factor and the personalization vector for the complete control problem. They are rather strict, which is why we then relax the problem to that of ranking control, where we can again obtain some bounds which are softer but still strict.

\subsubsection{Complete control}

The proper tuning of both the damping parameter $\alpha$ and personalization vector $\bv$ is essential for achieving the desired PageRank centrality in a given network (as discussed in \cite{boldi2005pagerank, garcia2012localization}). Our work in \cite{contrerasaso2023pagerank} established the relationship between these parameters, specifically examining which ranges of $\alpha$ yield the target centrality vector when combined with appropriate choices of $\bv$. This raises a fundamental question: given a graph and damping factor, can we achieve any desired PageRank vector by selecting the right personalization vector?

The answer proves to be negative, as there exist cases where no positive solution ($v_i>0,, \forall i$) can be found. However, we can investigate the conditions that determine whether $\bm{\pi}_0$ corresponds to some personalization vector $\bm{v}$. The following result characterizes when positive personalization vectors exist to produce a specified PageRank centrality $\bm{\pi}_0$.

\begin{theorem}[Existence of the personalization vector \cite{contrerasaso2023pagerank}]\label{prop-existencev}
	Given a graph $G$ and a positive, unit norm vector $\bm{\pi}_0$ then there exists a positive, unit norm personalization vector $\bm{v}$ such that $\bm{\pi}_0$ is the PageRank vector if and only if $\bm{\pi}_0^T \bm{e}_j > \alpha \bm{\pi}_0^T P \bm{e}_j$ for all $j$.
\end{theorem}

This result makes two points clear about the required damping factor: first, it will generally need to be small, if we expect those $N$ inequalities to be fulfilled at the same time. Second, the larger the network, the smaller it will need to be, not only for the amount of inequalities but also because of the increased size of $P$.

\subsubsection{The ranking control problem}\label{sec:ranking}

The previous conditions on the damping factor make it clear that, for real networks (whose size tends to be large), there will be absolutely no room for complete centrality control. However, for most applications the individual centrality scores are not relevant. Instead, what is relevant is the relative positioning of nodes in the overall ranking. A ranking among the nodes $V$ is a partial order $\preceq$ on $V$ such that for any $i,j\in V$, $i\preceq j$ means that $j$ is ranked higher than or equal to $i$.

When it comes to discussing controllability, this is a far less restrictive setting (the ``amount'' of centrality vectors is uncountably infinite, while the amount of rankings is finite). In order to study this problem we resorted to a change in viewpoint to a geometric one \cite{contrerasaso2023pagerank}, based on the fact that: consider the $N$-simplex defined as we can understand equation \eqref{eq:PRexplicit} as the following map
\begin{align}\label{eq-PRmap}
	\bm{\pi}(G, \alpha, \cdot):&\quad  \Delta_N \longrightarrow \Delta_N \nonumber \\
	&\quad\,\, \bm{v} \lhook\joinrel\longrightarrow \bm{\pi}(G, \alpha, \bm{v}),
\end{align}
where $\Delta_N = \left\{\bm{x} \in \R^N  \; | \; \bm{x}>0,\, ||\bm{x}||_1=1 \right\}$ is the positive $N$-simplex. This is because this is the space of all possible personalization and PageRank vectors of graphs with $N$ nodes. Since all rankings can be found around the center of the simplex, there is ranking control if and only if
\begin{equation}\label{eq-rankingcenter}
	\bm{e}_0= \frac{1}{N}\bm{e} \in \textrm{Im}(\bm{\pi}) \quad \text{and}\quad \bm{e}_0 = \frac{1}{N} \bm{e} \notin \partial \textrm{Im}(\bm{\pi}).
\end{equation}

Imposing that a personalization vector must exist such that this vector is in the image of PageRank, and utilizing Theorem~\ref{prop-existencev}, we arrive at the following theorem.

\begin{theorem}[Characterization of ranking control \cite{contrerasaso2023pagerank}]\label{thm:ranking}
	Given a graph $G$ and damping factor $\alpha=(0,1)$, then it is possible to obtain any ranking of the nodes under the PageRank if and only if
	\begin{equation}\label{eq-necessary}
		\frac{1}{\alpha} > \max_j\left(\sum_{i=1}^N P_{ij}\right).
	\end{equation}
\end{theorem}

Notice the reduction from $N$ inequalities (Theorem~\ref{prop-existencev}) to just one. This renders ranking control far more feasible than complete control. However, this result also fixes an upper bound for $\alpha$ based on the maximum total probability of visits to any node $j$. When it is fulfilled, any ranking is achievable through proper selection of the personalization vector. This is a deterministic constraint - even a single node with many low out-degree incoming connections severely limits ranking control, regardless of the broader network structure. This limitation particularly affects scale-free networks~\cite{barabasi1999emergence}, which characteristically contain such high in-degree nodes.

\paragraph{The bound on real networks}

Having established a network-specific upper bound for $\alpha$ that enables ranking control via the personalization vector, we determined its strictness \cite{contrerasaso2023pagerank}. While $\alpha = 0.85$ is standard \cite{langville2006google}, representing roughly 8 hyperlink clicks before reset, this implies a maximum column sum of $P$ around 1.17, quite a restrictive condition.

We analyzed various networks from the KONECT \cite{KONECT-datasets} and CASOS \cite{CASOS-datasets} repositories\footnote{For dangling nodes, we added one random connection to a non-dangling node as a minimal intervention.} to compute their maximum permissible $\alpha$ values for ranking control, shown in Figure \ref{fig:real-datasets-alpha}.

\begin{figure*}[ht!]
	\centering
	\includegraphics[width=0.9\textwidth]{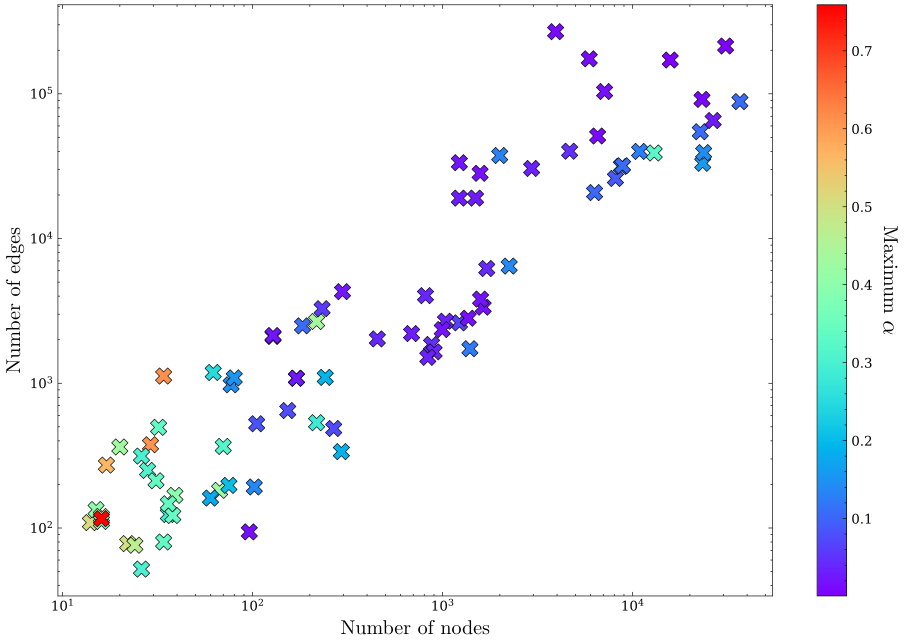}
	\caption[Scatter plot showing the number of edges against the number of nodes for real networks and their maximum $\alpha$ for ranking control.]{Scatter plot showing the number of edges against the number of nodes for 84 different real networks obtained from the KONECT network repository \cite{KONECT-datasets} and the CASOS network repository \cite{CASOS-datasets}, with datapoints colored based on the maximum value of $\alpha$ providing ranking control. Figure reproduced from \cite{contrerasaso2023pagerank} with permission.}
	\label{fig:real-datasets-alpha}
\end{figure*}

The results show maximum $\alpha$ values consistently below the standard 0.85, with smaller networks offering better controllability. This aligns with Theorem~\ref{thm:ranking}, as larger networks tend to have higher maximum column sums in $P$, reducing controllability.

These findings have significant implications for network analysis and information flow control. The demonstrated limitations on PageRank controllability, particularly in large-scale networks, challenge the conventional use of $\alpha = 0.85$ when specific ranking outcomes are desired. This is especially relevant for applications in search engine optimization, recommendation systems, and influence maximization strategies, where practitioners often seek to adjust node centralities. The results suggest that achieving desired ranking patterns may require substantially lower damping factors than traditionally used, particularly in large networks or those with heavy-tailed degree distributions. This creates a practical trade-off between ranking controllability and the local exploration of network structure that higher damping factors provide.

\subsection{Fully personalized PageRank: controlling node-dependent random walks}\label{subsec:paramRelatedcontrol}

As we briefly touched upon in Section~\ref{sec:preliminary}, there are two interesting variants of PageRank, which consider node-dependent dampings as well as node-dependent personalization vectors. It would be interesting to understand if the analysis carried out in the previous subsections can be extended to them, and if so, what changes.

Here we will briefly examine those two variants, to find analogous bounds which can then be used to judge how controllable they are, in the same spirit as in the previous Subsections, comparing them to the classical PageRank case.

\subsubsection{Ranking control of node-dependent restart PageRank}

We end this section examining how the rankings of the node-dependent restart PageRank, introduced in Section~\ref{sec:preliminary}, are affected by the choice of parameters, as in the previous cases (standard and biplex PageRank). The measure is again equipped with similar ingredients as the standard PageRank, as there is a collection of dampings $\alpha_i \in (0,1),\,i=1,\dots,N$ involved, as well as a personalization vector $\bw\in\R^N$.

\paragraph{Interlude: realization of node-dependent restart PageRank via standard PageRank}

Before analyzing the relation between these and the centralities, it is worth understanding how are they tied when compared to PageRank. In particular, we want to quantify how redundant is this measure as opposed to the vanilla PageRank, in the sense of achieving the same outcomes. This is relevant for control purposes, as we want to know how much more versatility is provided by this generalization.

\begin{theorem}
	Let $G=(V,E)$ be a graph, $\alpha_i\in(0,1),\, i\in V$ node-dependent dampings and $\bm{w}\in \mathbb{R}^N,\, \bm{w}>0,\,||\bm{w}||_1=1$ a personalization vector in the node-dependent restart PageRank $\bm{\pi}_{\rm NPR}$. There always exists a personalization vector $\bv\in\R^N,\,\bv>0,\,||\bv||_1=1$ such that the standard PageRank coincides with the node-dependent restart PageRank $\bm{\pi}_{\rm PR}=\bm{\pi}_{\rm NPR}$ if 
	\begin{equation}
		\bw_\cA^T \be_i > \alpha \bw_\cA^T P \be_i \quad  \forall i\in V,
	\end{equation}
	where $\bw_\cA^T=\bw^T(\I_N-\cA P)^{-1}$.    
\end{theorem}

\begin{proof}
	Note that from \eqref{eq:NDRPageRank} one can derive the following equation for the personalization vector $\bm{w}$ \cite{avrachenkov2014personalized}:
	\begin{equation}\label{eq:boldi-dampinig_vector}
		\bm{\pi}_{\rm NPR}(\cA,\bw) = \frac{1}{\gamma} \bw^T(\I_N - \cA P)^{-1},
	\end{equation}
	where
	\begin{equation}\label{eq:gamma_NDPR}
		\gamma = \bw^T(\I_N - \cA)^{-1}\be = \sum_{i=1}^N \frac{\bw_i}{1-\alpha_i}>1.
	\end{equation}
	
	The condition we therefore must check is 
	\begin{equation}
		\bm{\pi}_{\rm NPR}(\cA,\bw) = \bm{\pi}_{\rm STD}(\alpha,\bv) \Rightarrow \frac{1}{\gamma} \bw^T(\I_N - \cA P)^{-1} = (1-\alpha) \bv^T (\I_N - \alpha P)^{-1}. 
	\end{equation}
	
	We can then write $v^T$ explicitly as
	\begin{equation}
		\bv^T= \frac{1}{\gamma(1-\alpha)} \bw^T(\I_N-\cA P)^{-1} (\I_N - \alpha P).
	\end{equation}
	
	We need to prove that this expression is positive and consistent with unit-normalization. Starting with the latter,
	\begin{equation}
		\bv^T \be =\frac{1}{\gamma(1-\alpha)} \bw^T (\I_N-\cA P)^{-1} (1 - \alpha ) \be = \frac{1}{\gamma}\,\bw^T (\I_N-\cA P)^{-1} \be = 1.
	\end{equation}
	
	As for positivity,  we have
	\begin{equation}
		v_i = \bv^T \be_i = \frac{1}{\gamma(1-\alpha)} \bw^T(\I_N-\cA P)^{-1} (\I_N - \alpha P) \be_i > 0,\quad \forall i.
	\end{equation}
	
	Notice that $1-\alpha>0$ and $\gamma \geq 1$, therefore this condition can be written as
	\begin{equation}
		\bw^T(\I_N -\cA P)^{-1} \be_i > \alpha \bw^T(\I_N -\cA P)^{-1} P \be_i,\quad \forall i.
	\end{equation}
	
	Defining $\bw_\cA^T=\bw^T(\I_N-\cA P)^{-1}$, we have that $\bw_\cA^T>0$, and
	\begin{equation}
		\bw_\cA^T \be_i > \alpha \bw_\cA^T P \be_i = \alpha \sum_{j=1}^N (\bw_\cA)^T_j P_{ji},\quad \forall i,
	\end{equation}
	concluding the argument.
\end{proof}

This result shows that, if we can change both $\alpha$ and $\bv$, the standard PageRank can always find that value of the node-dependent restart PageRank. In hindsight, this rather sensible: in the limit $\alpha\rightarrow 0$, $\bv$ dictates the centralities. 

A more nuanced question is then the reverse one: can we conceive conceive a centrality $\bm{\pi}_0\in\R^N$, with $\bm{\pi}_0>0,\,||\bm{\pi}_0||_1=1$ which can't be realized for a given $\alpha$ in the standard PageRank (see Theorem \ref{prop-existencev} for this condition), which nevertheless can be obtained from higher values of the dampings $\alpha_i > \alpha$ in the node-dependent restart case, due to the additional freedom it provides? This is something we will soon discuss.

\paragraph{Localization of node-dependent restart PageRank}

We can derive similar results as those in \cite{garcia2012localization}, discussed in Subsection~\ref{subsec:PRlocalization}, for the node-dependent restart PageRank, regarding the localization of each PageRank scores and subsequently of the possible competitors. 

\begin{lemma}\label{lem:diagonaldom-NPR}
	Let $P$ be a row-stochastic matrix, $\alpha_i\in(0,1),\,\forall i$, $\cA=\textrm{diag}(\alpha_1,...,\alpha_N)$ and $\gamma = [\bw^T(\I_N -\cA)^{-1}\be]^{-1}$ with $\bw\in \mathbb{R}^{N},\, \bw>0,\, |\bw|_1=1$. Then
	\begin{itemize}
		\item $Y=\I_N -\cA P$ is strictly row-diagonally dominant.
		\item $X=\gamma\, Y^{-1}$ is strictly diagonally dominant of its column entries.
		\item The $i^{th}$ column of $X$ attains its maximum value at $x_{ii}$.
	\end{itemize}
\end{lemma}

\begin{proof}
	Firstly, as $P$ is row-stochastic ($P\be=\be$) the sum of each row of $Y$ is $(Y\be)_i=(\be-\cA P\be)_i=(\be-\cA \be)_i=1-\alpha_i$. Therefore, since $\alpha\in(0,1)$ and $0\leq p_{jk} \leq 1$ for all $i,k=1,...,N$ we get
	\begin{equation}
		|y_{ii}| = |1-\alpha_{i} p_{ii}| = 1 - \alpha_i p_{ii} = 1 - \alpha_i + \alpha_i \sum_{k\neq i} p_{ik} > \alpha_i \sum_{k\neq i} p_{ik} = \sum_{k\neq i} |y_{ik}|,
	\end{equation}
	where we used the fact that $P$ is row-stochastic. We have therefore shown that $Y$ is strictly row-diagonally dominant. By Theorem 2.5.12 of \cite{horn1991topics} we know that $Y^{-1}$ and $X$ are strictly diagonally dominant of their column entries. Therefore,
	\begin{equation}
		|x_{ii}| > |x_{ki}| \quad \forall i, k \neq i.
	\end{equation} 
	
	On the other hand, $Y$ is a non-singular $M$-matrix (see above, taking $s=1$, as the spectral radius of $\cA P$ is less than 1), thus $Y^{-1}\geq 0$ \cite{meyer2001matrix}. Hence, the absolute values of the above formula can be deleted and we have
	\begin{equation}
		\max_{k} x_{ki} = x_{ii},
	\end{equation}
	finishing the proof.
\end{proof}

Armed with this lemma, we can tackle the localization problem. 
\begin{definition}
	Given a graph $G=(V,E)$ and fixed node-dependent dampings $\alpha_i\in (0, 1),\, i= 1,...,N$, for each node $i \in V$ we define $\mathcal{P R}(i)$ as the set of all possible values of the node-dependent restart PageRank of node $i$, i.e.
	\begin{equation}
		\mathcal{P R}(i) = \{\bm{\pi}^T(\bw) \be_i\quad \text{for all}\quad \bw\in\mathbb{R}^N, \, \bw>0,\, \, ||\bw||_1=1\}.
	\end{equation}
\end{definition}

The concrete values of $\mathcal{PR}(i)$ can be quantified with the following theorem:

\begin{theorem}
	Under the same conditions and notation as the above definition,
	\begin{equation}
		\mathcal{P R}(i) = (\min_j x_{ji}, x_{ii}).
	\end{equation}
\end{theorem}

\begin{proof}
	There are two steps involved
	\begin{enumerate}
		\item We first want the upper and lower bounds. Without loss of generality, $i=1$. Then, $\bm{\pi}^T(\bw)\be_1 =\sum_{j=1}^N w_j x_{j1}$. Now, as $\bw>0$ and $||\bw||_1=1$, we have
		\begin{equation}
			\min_j x_{ji} < \sum_{j=1}^N v_j x_{j1} < \max_j x_{j1} = x_{ii},
		\end{equation}
		where the last equality is due to Lemma~\ref{lem:diagonaldom-NPR}.
		
		\item We now want to see that all values within the interval can be found with suitable $\bw$. Again without loss of generality $i=1$. Define the convex combination vector
		\begin{equation}
			\bw^\epsilon_\lambda = \lambda \bw_1^\epsilon + (1-\lambda) \bw_{j_1}^\epsilon > 0, \quad \lambda\in (0,1),
		\end{equation}
		where $j_1$ is the minimum of the first column of $X$ and
		\begin{align}
			\bw_1^\epsilon &= \left(1-\epsilon, \frac{\epsilon}{n-1}, \frac{\epsilon}{n-1}, ..., \frac{\epsilon}{n-1}\right)^T, \nn \\
			\bw_{j_1}^\epsilon &= \bigg(\frac{\epsilon}{n-1}, ..., \underbrace{1-\epsilon}_{j_1}, ..., \frac{\epsilon}{n-1}\bigg)^T. 
		\end{align}
		
		This vector satisfies
		\begin{equation}
			\lim_{\lambda\rightarrow1} \lim_{\epsilon\rightarrow0} \bm{\pi}^T(\bw_{\lambda}^\epsilon)\be_1 = x_{11},\quad \lim_{\lambda\rightarrow0} \lim_{\epsilon\rightarrow0} \bm{\pi}^T(\bw_{\lambda}^\epsilon)\be_1 = x_{j_1 1}.
		\end{equation}
		Hence, for every $x$ with $x_{j_1 1} < x < x_{11}$ there exists some $\epsilon_0,\lambda_0 \in (0,1)$ such that $\bm{\pi}^T(\bw_{\lambda_0}^{\epsilon_0})\be_1 = x$.
	\end{enumerate}
\end{proof}

It is clear from this result that the localization of the node-dependent restart PageRank is analogous to the standard one, the difference is in the actual value of the matrix $X$ used to compute that range. 

This, in turn, could be used to analyze the competitors and leaders in any network, as defined in \cite{garcia2012localization}, since they only require knowledge of the $\mathcal{PR}(i)$ ranges in order to find them.

\paragraph{Ranking control of node-dependent restart PageRank}

Here we will pick up the discussion of the possible centralities realized with the node-dependent restart PageRank, as compared to the standard PageRank. We are interested in examining if there is a trade-off between having several, node-dependent dampings (which provides more freedom in the measure's parameters) and the actual values of these dampings.

For instance, one could wonder if lowering substantially the values of certain dampings could compensate for higher ones (which restrict the random walk to follow the network structure more closely), allowing for an enhance in controllability. However, as we will see, this is not the case: in fact, the node-dependent restart PageRank is  constrained by the standard one.

We will follow closely the narrative from Subsection~\ref{subsec:paramPRcontrol}. Starting from \eqref{eq:NDRPageRank} we can derive the following formula for the personalization vector $\bw\in\R^N$ as a function of the centrality $\bm{\pi}\in\R^N$
\begin{equation}
	\bw^T=\frac{1}{\bm{\pi}^T(\I_N-\cA )\be} \bm{\pi}^T(\I_N - \cA P).
\end{equation}

We can use this formula to derive the following bound, akin to Theorem \ref{prop-existencev}.
\begin{theorem}[Existence of the personalization vector]\label{thm:prop-existencev-ndpr}
	Given a graph $G=(V,E)$ and a positive, unit norm vector $\bm{\pi}_0$, then there exists a positive, unit norm personalization vector $\bw\in\R^N$ such that the node-dependent restart PageRank is $\pi_{\rm NPR}=\pi_0$ if and only if $\bm{\pi}^T_0 \be_j > (\max_i{\alpha_i}) \bm{\pi}^T_0 P \be_j$ for all $j$.
\end{theorem}

\begin{proof}
	We need to show under which conditions does $\bw$ have unit norm and positivity, for it to be a personalization vector. 
	
	First we check the unit-norm, 
	\begin{equation}
		||\bw||_1=\bw^T\be = \frac{1}{\bm{\pi}_0^T(\I_N-\cA )\be} \bm{\pi}_0^T(\I_N - \cA P)\be = ||\bm{\pi}_0||_1=1.
	\end{equation}
	where we used the row-stochasticity of $P$, i.e. $P\be=\be$. As for positivity,
	\begin{align}
		\bw_j = \bw \be_j &= \frac{1}{\bm{\pi}_0^T(\I_N-\cA )\be} \bm{\pi}_0^T(\I_N - \cA P)\be_j \nn \\ &>\frac{1}{\bm{\pi}_0^T(\I_N-\cA )\be} \bm{\pi}_0^T(\I_N - (\max_i{\alpha}_i) P)\be_j > 0,
	\end{align}
	which completes the proof.
\end{proof}

For the ranking control, things are slightly more inconvenient in this case, compared to the standard PageRank. In the closed-form formula for the PageRank vector the normalization is independent of the personalization vector $\bv$ used, it is just $1-\alpha$. Here, however, the normalization \eqref{eq:gamma_NDPR} in \eqref{eq:boldi-dampinig_vector} depends on the choice of personalization vector $\bw$ under consideration. It is therefore interesting for us to leave the normalization behind, and consider the following map:
\begin{align}
	\widetilde{\bm{\pi}}(\cA,\cdot) =&\quad \triangle_N \longrightarrow CH_N \nn \\
	&\quad\,\,\,\, \bw \lhook\joinrel\longrightarrow \bw^T(\I_N -\cA P)^{-1},
\end{align}
where $CH_N$ stands for the positive convex hull in $N$ dimensions 
\begin{equation}
	CH_N = \left\{\bx\in \mathbb{R}^N \;\bigg|\; \bx = \sum_{i=1}^N{a_i \be_i}, \, a_i>0  \right\}.
\end{equation}

This map is linear in $\bw$, therefore it maps the $N$-dimensional ``normalized'' simplex $\triangle_N$ into an $N$-dimensional simplex (possibly un-normalized) in the convex hull. In this new setting, the generalization of Theorem \ref{thm:ranking} requires the existence of a personalization vector such that $\widetilde{\bm{\pi}}_0$ is in the parameterized line centered in the convex hull.

\begin{theorem}
	Given a graph $G$ and a set of damping factors $\alpha_i=(0,1),\, i\in V$, then it is possible to obtain any ranking of the nodes under the node-dependent restart PageRank if and only if
	\begin{equation}
		\frac{1}{\max_k \alpha_k} >  \max_j \left(\sum_i^N P_{ij} \right),
	\end{equation}
\end{theorem}

\begin{proof}
	Firstly, the relation between $\widetilde{\bm{\pi}}$ and $\bm{\pi}_{\rm NDR}$ is simply a proportionality factor, which therefore can't change the relative ranking between the components. 
	
	Second, the condition for the existence of ranking control is for the line parameterized by $\tilde{\be}=(t,t,\dots,t)=t\, \be \in\R^N$ with $t>0$ to pass through the image of $\widetilde{\bm{\pi}}$. 
	
	Imposing this in Theorem \ref{thm:prop-existencev-ndpr} we easily arrive to the set of restrictions $1 > \be^T \cA P \be_j$ for all $j$. We can restrict it further taking maximums as 
	\begin{equation}
		1 > \be^T \cA P \be_j \geq (\max_k \alpha_k)  \left(\sum_i^N P_{ij}\right), \,\forall j\quad \Rightarrow\quad  1 \geq (\max_k \alpha_k)  \max_j\left(\sum_i^N P_{ij}\right),
	\end{equation}
	thus concluding the proof.
\end{proof}

The punchline of this calculation is that, actually, the freedom in having a collection of damping factors rather than a global one does not guarantee an improvement in terms of ranking control with respect to the standard PageRank: as a matter of fact, the standard PageRank with the damping factor equal to highest one from the collection (hence, the most ``restrictive'' one) has the same ranking control inequality.

\subsubsection{PageRank with a personalization matrix} \label{subsec:persmatrix}

We now move on to discuss the controllability of the personalization matrix PageRank, defined in \eqref{eq:PMPageRank}. One would naïvely expect that the generalization, with its corresponding increase in the degrees of freedom, would imply for the measure an increase in the possibilities for control and/or localization of the centrality scores. This is, however, not the case.

The following result will be key to understanding the relation between the standard and matrix personalization PageRank.
\begin{theorem}\label{thm:matrixPR}
	Let $G=(V,E)$ be a graph with $N$ nodes, let $\alpha\in(0,1)$ and $M\in\R^{N\times N}$ a row-normalized matrix. Let $\bm{\pi}_{\rm PR}(G,\alpha,\bv)$ be the standard PageRank and let $\widetilde{\bm{\pi}}$ be the result of the matrix personalized PageRank with damping factor $\alpha$ and personalization matrix $M$. We have
	\begin{equation}
		\widetilde{\bm{\pi}} = \bm{\pi}_{\rm PR}(G, \alpha, \widetilde{\bm{\pi}}^T M).
	\end{equation}
\end{theorem}
\begin{proof}
	Starting with the definition of the matrix personalized PageRank, we have
	\begin{equation}
		\widetilde{\bm{\pi}}^T (\alpha P + (1-\alpha) M) = \widetilde{\bm{\pi}}^T \Rightarrow \widetilde{\bm{\pi}}^T \alpha P + (1-\alpha) \widetilde{\bm{\pi}}^T M = \widetilde{\bm{\pi}}^T.
	\end{equation}
	
	The identification $\bv^T = \widetilde{\bm{\pi}}^T M$, knowing that $\widetilde{\bm{\pi}}^T \be =1$, results in
	\begin{equation}
		\widetilde{\bm{\pi}}^T (\alpha P + (1-\alpha) \be \cdot \bv^T) = \widetilde{\bm{\pi}}^T,
	\end{equation}
	which is the definition of the standard PageRank $\bm{\pi}_{\rm PR}$ itself.
\end{proof}

The implication of the above theorem is rather extraordinary: for a fixed $\alpha$, even when we enable the possibility of specifying individual personalization vectors per node, at the end of the day the final PageRank outcome could have been produced with a suitably chosen personalization vector, identical for all nodes. This might hint at a reason why the matrix personalization case has not been considered in the literature thus far.

Note that this does not mean that the matrix personalized PageRank is useless: while it would be possible to interpret the personalization matrix (each row having the distribution of teleportation probabilities corresponding to that node), a priori there is no interpretation for the personalization vector matching the PageRank outcome. However, what is formally shown here is the fact that we can obtain any personalization matrix PageRank outcome, without an actual personalization matrix, but just some personalization vector.

It is important to remark that we also have the converse implication: any standard PageRank outcome $\widetilde{\bm{\pi}}$ can be realized via the matrix personalization PageRank (simply considering $M=\be\cdot \bv^T$).

The following two corollaries follow directly from the previous result.
\begin{corollary}
Let $G=(V,E)$ be a graph, $\alpha\in(0,1)$ a fixed damping factor. Let $\mathcal{PR}(i), \mathcal{MPR}(i)$ the set of all possible centrality values of node $i\in V$ according to the standard and matrix personalization PageRank, respectively. Then, $\mathcal{MPR}(i) = \mathcal{PR}(i),\forall i\in V$.	
\end{corollary}

\begin{corollary}
	Let $G=(V,E)$ be a graph, $\alpha\in(0,1)$ a fixed damping factor. Then, there is full or ranking control in the matrix personalization PageRank if there is full or ranking control in the standard PageRank, respectively. 
\end{corollary}

These results are clear from the fact that all PageRank outcomes can be realized as matrix personalization PageRank outcomes, and viceversa, and in turn imply that the discussion on competitors and leaders under this centrality measure is equivalent to the one in \cite{garcia2012localization}, since the ranges are identical.

\section{Conclusions}\label{sec:conclusions}

Our comprehensive investigation into the controllability of PageRank-based centrality measures has revealed fundamental insights into the complex relationship between network structure, stochastic processes, and node importance rankings. Through detailed analysis of both structural and parametric control mechanisms, we have established clear boundaries for what can be achieved in terms of manipulating these measures, while preserving their essential character as representations of random walk processes.

Our analysis demonstrates that complete control - the ability to achieve arbitrary centrality scores - is only possible within specific parameter regions and network configurations that maintain the underlying Markov process properties. This is true for both the standard PageRank as well as the node-dependent generalizations. The weaker notion of ranking control \cite{arrigo2020beyond, contrerasaso2023pagerank}, while more achievable, still faces fundamental limitations rooted in the stochastic nature of the measures.

Our characterization of competitors - nodes whose relative rankings can be adjusted through parameter variations - and leaders - nodes capable of achieving maximum centrality within certain parameter regions - provides a nuanced understanding of how modifications to the underlying random walk can affect node importance hierarchies. These findings reveal the delicate balance between network structure and parameter choices in determining centrality outcomes, highlighting how the stochastic foundation of these measures interacts with attempts to control their results.

The established bounds on parameter regions for both complete and ranking control offer practical guidance for network analysis and design. These bounds are particularly relevant in applications where understanding the limitations of centrality manipulation is crucial, such as in the design of robust rankings for search engines or the analysis of influence spreading in social networks \cite{gleich2015pagerank}. Future research directions could explore how these controllability results extend to other spectral centrality measures and more complex stochastic processes on networks.

\section*{Acknowledgments}

The authors would like to thank Julio Flores and Esther García for fruitful conversations on the matrix personalized PageRank. G. C-A. is funded by the URJC fellowship PREDOC-21-026-2164. This work has been partially supported by project M3707 (URJC Grant) and the INCIBE/URJC Agreement M3386/2024/0031/001 within the framework of
the Recovery, Transformation and Resilience Plan funds of the European Union (Next Generation EU).


%
%
%

	\bibliography{centrality, control, datasets, standard, textbooks, hypergraphs, extra}

\end{document}